%% file: main.tex
\documentclass[a4paper,superscriptaddress,prl,notitlepage,twocolumn,floatfix]{revtex4-2}
\usepackage{amsmath}
\usepackage{amssymb}
\usepackage{amsthm}
\usepackage{graphicx}
\usepackage{xcolor}
\usepackage[english]{babel}
\usepackage{csquotes}
\usepackage{braket}
\usepackage{hyperref}
\usepackage{mathtools}
\usepackage{tikz}
\usepackage{times}
\usepackage{placeins}
\usetikzlibrary{quantikz}

\input{commands.tex}

\newcommand{\ie}{\textit{i.e.}\ }

\newtheorem{theorem}{Theorem}
\newtheorem{lemma}{Lemma}
\newtheorem{proposition}{Proposition}
\newtheorem{definition}{Definition}
\newtheorem{corollary}{Corollary}

\newtheoremstyle{propositionnum}{\topsep}{\topsep}{\itshape}{}{\bfseries}{.}{ }{\thmname{#1}\thmnote{ \bfseries #3}}
\theoremstyle{propositionnum}

\renewcommand{\aa}{\boldsymbol{a}}

\newcommand{\fu}{Dahlem Center for Complex Quantum Systems, Freie Universit\"{a}t Berlin, 14195 Berlin, Germany}
\newcommand{\hzb}{Helmholtz-Zentrum Berlin f{\"u}r Materialien und Energie, 14109 Berlin, Germany}
\newcommand{\hhi}{Fraunhofer Heinrich Hertz Institute, 10587 Berlin, Germany}

\begin{document}	
	\title{Classical surrogates for quantum learning models}
	\date{\today}
	\author{Franz J. Schreiber}
    \affiliation{\fu}
	
	\author{Jens Eisert}
    \affiliation{\fu}
    \affiliation{\hzb}
    \affiliation{\hhi}
	
	\author{Johannes Jakob Meyer}
    \affiliation{\fu}
	
	\begin{abstract}
    	The advent of noisy intermediate-scale quantum computers has put the search for possible applications to the forefront of quantum information science. 
	    One area where hopes for an advantage through near-term quantum computers are high is quantum machine learning, where variational quantum learning models based on parametrized quantum circuits are discussed. 
	    In this work, we introduce the concept of a \emph{classical surrogate}, a classical model which can be efficiently obtained from a trained quantum learning model and reproduces its input-output relations. 
	    As inference can be performed classically, the existence of a classical surrogate greatly enhances the applicability of a quantum learning strategy.
	    However, the classical surrogate also challenges possible advantages of quantum schemes. 
	    As it is possible to directly optimize the ansatz of the classical surrogate, they create a natural benchmark the quantum model has to outperform.
	    We show that large classes of well-analyzed re-uploading models have a classical surrogate. 
	    We conducted numerical experiments and found that these quantum models show no advantage in performance or trainability in the problems we analyze. 	    
	    This leaves only generalization capability as possible point of quantum advantage and emphasizes the dire need for a  better understanding of inductive biases of quantum learning models.	    
	\end{abstract}
	
	\maketitle

\emph{Quantum machine learning (QML)} is a popular and widely studied application of quantum computers~\cite{schuld2015introduction,biamonte2017quantum,bharti2021noisy}. Theoretical evidence suggests that one day quantum machine learning methods can outperform classical computers in certain classical~\cite{sweke2021quantum,liu2021rigorous} and quantum learning tasks~\cite{huang2021QuantumAdvantageExperiments}. Besides using quantum algorithms to train classical models~\cite{rebentrost2014quantum,verdon2018universal,hubregtsen2020integration}, a particular emphasis is put on the construction of \emph{quantum learning models}, which use quantum computers to parametrize hypothesis classes that are fit to the training data. 
Recently, much work has been done exploring \emph{variational} models that use a \emph{parametrized quantum circuit (PQC)} to make predictions, also referred to as \emph{quantum neural networks (QNNs)}. 
While variational quantum models can be implemented on today's \emph{noisy intermediate-scale quantum (NISQ)} devices~\cite{bharti2021noisy, Tyson2019HamSpectra, Lubasch2020nonlinear}, it is not clear if and how a practical quantum advantage can be realized within this framework~\cite{schuld2022is}. 

An especially pressing issue when dealing with quantum learning models is the reliance on quantum hardware which severely limits how such models can be deployed in production environments.
In this work, we argue that this challenge can be addressed if the quantum learning model in question has a \emph{classical surrogate}, which we define as an equivalent classical model that can be obtained efficiently from a trained quantum learning model (see Fig.~\ref{fig:front_figure}). The existence of a classical surrogate is a strong feature of a quantum learning model and can be considered as a fundamental prerequisite for any quantum learning model to be considered \enquote{practical}. We show that a type of variational quantum re-uploading models~\cite{perez-salinas2020data} considered in Refs.~\cite{vidal2018calculus,vidal2020input,schuld2021effect,caro2021encoding-dependent} admits classical surrogates. 
\begin{figure}
    \centering
    \includegraphics[]{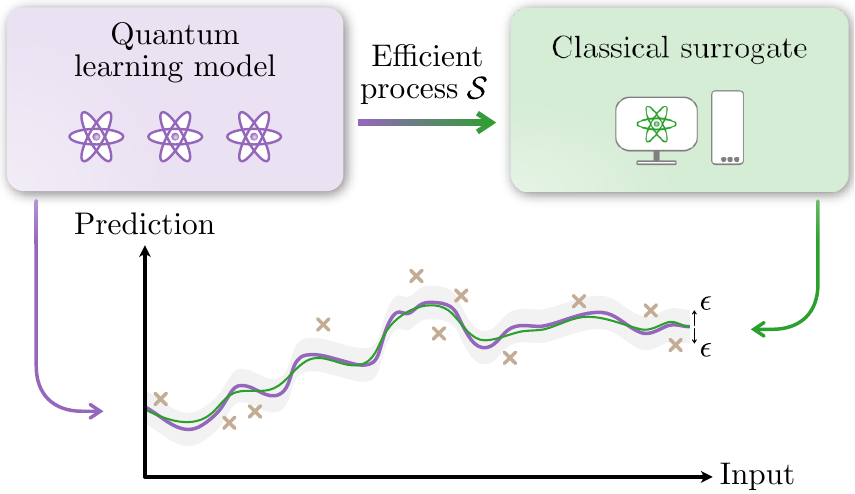}
    \caption{A quantum learning model has a classical surrogate if there exists a process that produces an equivalent classical model that is both efficient in the size of the quantum learning model and the desired approximation parameters.}
    \label{fig:front_figure}
\end{figure}

If a quantum model has a classical surrogate, quantum hardware is only required at the training stage. Any sort of quantum advantage therefore has to materialize at this stage through better training performance or generalization capability. 
In this setting the classical surrogate can also create a natural benchmark that the quantum model needs to beat to be relevant. This is the case when the classical surrogate itself can be turned into a learning model in its own right. We show that this is indeed the case for the re-uploading models considered in this work and numerically compare them with their classical surrogates on selected datasets. 

In our experiments, we are unable to observe any advantage of the quantum learning model, neither in performance, training or generalization. We show that the classical surrogate can always achieve lower training loss and has a more favorable optimization landscape but we can not generally rule out the existence of an advantage in generalization capability. However, such an advantage would necessitate an understanding of the implementation of suitable inductive biases~\cite{kubler2021inductive} in quantum learning models which is beyond the current state of the field.

\emph{Classical surrogates.}
Learning models executed on quantum computers come with the inbuilt reliance on quantum hardware. This can impede their practicality tremendously -- especially in the current NISQ-era where access to quantum computing resources is scarce. 
This situation is diametrically opposed to what makes classical machine learning attractive, where training the model might be very demanding, but obtaining new predictions is simple and can be done on less powerful client-side devices.

We can circumvent this impediment if we have access to a classical model that reproduces the same input-output relations as the quantum learning model. 
We call such a classical replacement of the trained quantum learning model a \emph{classical surrogate}. 
A essential prerequisite for such a definition to be non-trivial is \emph{efficiency} in the process that creates the surrogate from the trained quantum learning model as well as in the evaluation and storage of the classical surrogate itself.
We can not expect that classical surrogates generically exist for all quantum learning models. 
Due to the concerns outlined above, having a classical surrogate is thus a fundamental property of any quantum learning model that can be considered \enquote{practical}.

Formally, we define a classical surrogate for a hypothesis class $\calF$ of quantum learning models with inputs $\xx \in \calX$ and outputs $y \in \calY$ as follows:
\begin{definition}[Classical surrogate]\label{def:classical_surrogate}
A hypothesis class of quantum learning models $\calF$ has \emph{classical surrogates} if there exists a process $\calS$ that upon input of a learning model $f \in \calF$ produces a classical model $g_f \in \calG$ such that the maximal deviation of the surrogate from the quantum learning model is bounded with high probability. Formally, we require
\begin{align}
    \bbP\left[\,\sup_{\xx \in \calX} \lVert f(\xx) - g_f(\xx)\rVert \leq \epsilon \, \right] \geq 1 - \delta,
\end{align}
for a suitable norm on the output space $\calY$. The surrogation process $\calS$ must me efficient in the size of the quantum learning model, the error bound $\epsilon$ and the failure probability $\delta$. 
\end{definition}
We consider the supremum norm for the deviation to be a necessary feature, as more coarse-grained notions of approximation could tolerate isolated \enquote{outliers} for which the surrogate produces very different outputs than the quantum model.

Interestingly, the existence of a classical surrogate immediately implies that any advantage of the quantum learning model must be realized at the training stage, either through significant speedups, increases in training performance, cost reductions or better generalization capabilities. For such advantages, the classical surrogate can provide a natural benchmark if it can be turned into a learning model in its own right. In this way, quantum learning models with classical surrogates can be amenable to a notion of \enquote{dequantization}.

\emph{Variational re-uploading models.}
We consider the same type of variational quantum learning model as Refs.~\cite{vidal2018calculus,vidal2020input,schuld2021effect,caro2021encoding-dependent} where vector-valued inputs \smash{$\calX = \bbR^d$} are mapped to real outputs $\calY = \bbR$. The model is constructed by applying $L$ layers of trainable unitaries \smash{$W^{(j)}(\ttheta)$} interleaved with data-encoding blocks \smash{$S^{(j)}(\xx)$} resulting in a parametrized circuit 
\begin{align} \label{eqn:qm}
    U(\xx, \ttheta) = W^{(L)}(\ttheta) S^{(L)}(\xx)  \dots W^{(1)}(\ttheta) S^{(1)}(\xx) W^{(0)}.
\end{align}
Predictions are obtained by evaluating the expectation value of an arbitrary observable $M$ with bounded operator norm after $U(\xx, \ttheta)$ is applied to the all-zero state:
\begin{align} \label{eqn:measure_qm}
    f_{\ttheta}(\xx) = \langle 0 | U^{\dagger}(\xx, \ttheta) M U(\xx, \ttheta) |0\rangle.
\end{align}
In Refs.~\cite{vidal2020input,schuld2021effect}, it has been shown that learning models of this type can be expanded into a truncated Fourier series
\begin{align}\label{eqn:fourier_expand}
    f_{\ttheta}(\xx) = \sum_{\oomega \in \Omega} c_{\oomega}(\ttheta) e^{-i \oomega \xx}, 
\end{align}
where the set of accessible frequencies $\Omega$ depends only on the structure of the \smash{$S^{(j)}(\xx)$} and the number of layers.

In the following, we will assume that the data encodings \smash{$S^{(j)}(\xx)$} are composed of elementary data encodings of the form \smash{$S_k^{(j)}(x_k) = \exp(-i x_k H_k^{(j)})$} with \smash{$H_k^{(j)}$} having integer eigenvalue differences, which means that all data features are elementary parameters of rotation gates~\footnote{Through this we isolate the properties of the underlying quantum model in our analysis. If pre-processing needs to be considered, then one could always use the different functions of the input variables as new parameters.}. The integer eigenvalue differences guarantee that the model output is periodic on the interval $[0,2\pi)$.
It has been shown in Ref.~\cite{caro2021encoding-dependent} that for models of this type with constrained locality of the gates \smash{$S_k^{(j)}(x_k)$}, the number of accessible frequencies as well as the maximal frequency grows only polynomially in the
number of encoding gates $N$ and hence also at most polynomial in the number of qubits. 

\emph{Fourier-based classical surrogates.}
We can explicitly exploit the fact that the outputs of the model are guaranteed to be truncated Fourier series with known frequencies to construct a classical surrogate for these models. We denote the Fourier-based surrogate as
\begin{align}\label{eqn:surrogate_fourier}
    g_{\cc}(\xx) = \sum_{\oomega \in \Omega} c_{\oomega} e^{-i \oomega \xx}
\end{align}
where the Fourier coefficients $\cc = (c_{\oomega})_{\oomega\in \Omega}$ are the parameters that need to be computed with a guarantee that fulfills the surrogation conditions of Definition~\ref{def:classical_surrogate}.

To do so, we use the following protocol based on the discrete Fourier transform: 
For each of the $d$ data features, set $T_i = 2 \omega_{\max}(i) + 1$ where $\omega_{\max}(i) = \max \{ |\omega_i | \colon \oomega \in \Omega  \}$ is the maximal frequency for the $i$-th data feature. We use this to define an equally-spaced grid on the interval $[0, 2\pi)$ with $T_i$ points for every data feature which yields a grid for the whole set with \smash{$T = \prod_{i=1}^d T_i$} elements. For every datapoint $\xx_j$ in this grid, we obtain $N$ samples from the quantum model's output and compute an estimate for the expectation value through the corresponding sample mean $\hat{y}_j$. We then solve the least-squares problem 
\begin{align} \label{eqn:least_square}
    \cc_{*} &= \operatornamewithlimits{argmin}_{\cc} \sum_{j=1}^T \big| g_{\cc}(\xx_j) - \hat{y}_j  \big|^2.
\end{align}
We transform this into a linear system by defining
\begin{align}
    A_{j,\oomega} = e^{-i \oomega \xx_j}, \,
    \hat\yy_j = y_j,
\end{align}
with which it reduces to
\begin{align} \label{eqn:linear_least_square}
    \cc_{*} &= \operatornamewithlimits{argmin}_{\cc} \lVert A \cc - \hat\yy\rVert^2.
\end{align}
The solution for this problem can be computed via the pseudo-inverse using a singular value decomposition or as a convex program. 
We have the following recovery guarantee for the whole protocol:
\begin{proposition}\label{prop:recovery_guarantees}
    The classical surrogate $g_{\cc_{*}}$ obtained through the above protocol fulfills the surrogation conditions of Definition~\ref{def:classical_surrogate} using a total of
    \begin{align}
        N_{\mathrm{total}} = T N = \frac{2 T \lVert M \rVert_{\infty}^2}{\epsilon^2}\left( \log \frac{1}{\delta} + T \log 2 \right) 
    \end{align}
    invocations of the quantum learning model.
\end{proposition}
\begin{proof}
The full proof is relegated to Appendix~\ref{sec:app_recovery_guarantees} for brevity. It leverages transportation-cost inequalities to show concentration of the $\ell_1$-norm approximation of the Fourier coefficients. The result follows from combining this with known guarantees for the discrete Fourier transform.
\end{proof}
The surrogation protocol only has a sub-linear overhead in $T$ compared to the sample complexity of conducting inference with the quantum model as obtaining the output for $T$ different inputs to accuracy $\epsilon$ with probability at least $1-\delta$ has a sample complexity of
\begin{align}
    N_{\mathrm{inference}} = \frac{2  T  \lVert M \rVert_{\infty}^2}{\epsilon^2}\log\frac{2 T}{\delta}
\end{align}
We provide a proof of this well-known fact in Appendix~\ref{sec:app_multiple_observables} for completeness.

As \smash{$T = O(\omega_{\max}^{d/2})$}, which is polynomial in the number of qubits, the above protocol is efficient in all relevant variables, as required by the surrogation conditions. While the number of data features $d$ is constant for a given learning problem, the exponential scaling in this variable can present a challenge to scale up the classical surrogate. It is an intriguing question how further structural assumptions could be used to improve upon the above protocol.

We now outline how we can directly train the classical surrogates as learning models. Nothing prevents us from directly minimizing the least-squares loss of Eq.~\eqref{eqn:least_square} for the given training data. We could use the same strategy as for the construction of the classical surrogate and solve the linear system through a singular value decomposition. For larger problems this, however, becomes impractical in time and memory requirements and the matrix $A$ could be ill-conditioned as we discuss in Appendix~\ref{sec:app_additional_info_numerics}. Furthermore, a perfect solution to the problem is usually not desirable to avoid overfitting.

We therefore opt to use stochastic gradient descent methods on the Fourier coefficients, which avoids the aforementioned problems. The convexity of the optimization problem furthermore guarantees convergence to the global optimum if the learning rate is suitably parametrized~\cite{Bottou2012}. We implement this model as a neural network which facilitates backpropagation through the surrogate model, which means we can apply the same optimization techniques to quantum model and classical surrogate.
The unrestricted optimization of the Fourier coefficient implies that the global optimum of the classical surrogate provides a lower bound for the global training loss achievable by the quantum model, but also means that there is a higher danger of overfitting which we mitigate by observing the validation loss. Additionally, the loss landscapes of the quantum model and the classical surrogate can differ dramatically because the linear least-squares problem is convex, whereas the loss landscapes of quantum learning models are usually rugged and complicated~\cite{rudolph2021orqviz}, especially at low parameter counts before overparametrization phenomena kick in~\cite{larocca2021theory,anschuetz2022critical}.

\emph{Numerical implementation.}
In this section, we compare re-uploading models based on parametrized quantum circuits to their classical surrogates. We emphasize that there are countless possibilities to tweak the performance of the classical surrogate, like regularization or reparametrization. The point of this section, however, is to find the simplest model that matches or beats the corresponding quantum learning model.

The quantum model we consider is a re-uploading model where each layer of the model consists of a data encoding block \smash{$S(\xx) = \bigotimes_{i=1}^d R_X(x_i)$} and a trainable block \smash{$W^{(l)}$}, $L \in \{0,\dots, L\}$. 
For the trainable block \smash{$W^{(l)}$} we choose the \textit{Strongly Entangling Layer} template provided by \textit{PennyLane} which consists of $B$ block layers, where $B$ -- as well as the number of total layers $L$ -- is a hyperparameter of the quantum learning model (see Appendix~\ref{sec:app_trainable_block}). In the following we use models with $B \in \{1,3\}$ and $L \in \{2,3\}$. 
The surrogate model is implemented as a neural network with one linear layer with the Fourier coefficients as weights. To avoid dealing with complex numbers, we have used the equivalent expansion in terms of cosines and sines. Note that the set of accessible frequencies only depends on $L$ and that increasing $B$ only increases the expressivity of the Fourier coefficients.

In our implementation, we use \mbox{\textit{PennyLane}}~\cite{Pennylane} for the quantum parts and \mbox{\textit{PyTorch}}~\cite{Pytorch} for the classical parts. The problems considered here are all of moderate size, therefore we use the memory-intensive LBFGS algorithm which is guaranteed to converge under the Wolf conditions, which spares the search for suitable learning rates \cite{Liu1989}.

\begin{figure}
    \centering
    \includegraphics[width=0.505\textwidth]{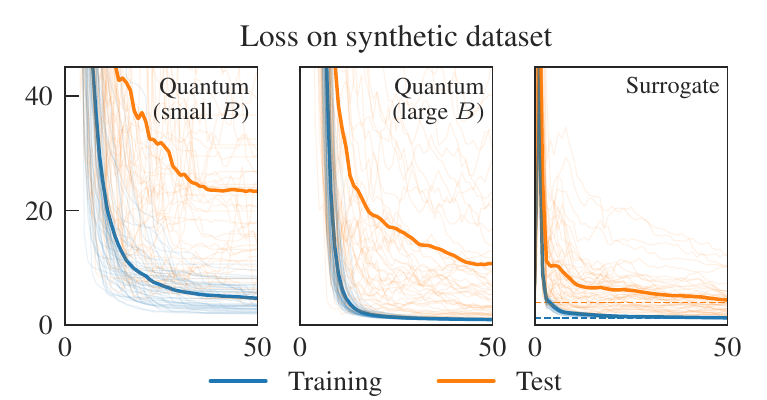}
    \caption{Training and test loss over epochs for three different learning models trained on a standard synthetic dataset over random splits. The dashed lines indicate the mean loss associated with the empirical risk minimizer. We observe the effects of higher expressivity when the number of parameters of the quantum model is increased but the surrogate model outperforms both quantum models in loss and trainability.}
    \label{fig:synthetic}
\end{figure}

\emph{Results.}
As we have already outlined, advantages of a quantum learning model could come in many flavors, be it easier trainability, better generalization or significant speedups. We compare the performance of quantum learning models with two different numbers of parameters and the corresponding classical surrogate on three different learning problems: A synthetic dataset generated using the \mbox{\textit{make\_regression}} function provided by \mbox{\textit{Scikit-learn}}~\cite{scikit-learn}, a dataset obtained from sampling outputs of a randomized quantum re-uploading model and the California housing dataset which is a standard benchmark dataset for regression tasks. We present the loss curves for the synthetic dataset in Fig.~\ref{fig:synthetic}, the loss curves for the other problems can be found in Appendix~\ref{sec:app_additional_figures}. As we used a small dataset for the synthetic case, we randomized the train-test-split for each run to avoid results that depend on a particular split.

As expected, we observe that the classical surrogate consistently achieves lower training loss as the quantum models and that the training loss converges as nicely as one would expect from a convex problem. We further witness that an increase in the number of block layers $B$, and hence a greater expressivity, allows the quantum model to perform better, both in test and training loss as well as in trainability in the synthetic and random PQC learning tasks. The increased performance in trainability is observed by the lower variance and increased smoothness of training curves. This indicates that the loss landscapes at low expressivity are highly frustrated which is in good accordance with recent results on loss landscape of variational models~\cite{wiersema2020exploring,rudolph2021orqviz,anschuetz2022critical,larocca2021theory}.

This observation, however, also points to the fact that with increased expressivity, the expected behavior of the quantum models also tends to be more and more similar to the one of the classical surrogate, which -- in a way -- presents an idealized limit of the quantum model. 
However, while this is desirable from the perspective of training performance, it also increases the danger that the quantum model looses its ability to encode an inductive bias that is different from the one of the classical surrogate in a meaningful way, in which case the direct optimization of the surrogate is usually the better alternative. 
This immediately raises the question if there even exists a \enquote{sweet spot} where quantum re-uploading models of the type considered in this work are advantageous, as they interpolate between a setting of high bias but low trainability and a setting of low bias and high trainability. 

The fact that we do not observe any kind of advantage of the quantum learning model over the classical surrogate in the examples we study can well be a limitation of the particular parametrizations of the quantum learning models that we consider, which are built on circuit templates available in the literature. We can therefore not make a statement about the ultimate capabilities of these models, as for this it is first necessary to better understand the relation of circuit structures and the corresponding inductive biases of the quantum learning model. However, we can conclude that for quantum learning models constructed from contemporary circuit ansätze, engineering an equivalent or better classical model is rather simple.

\emph{Conclusion.}
In this work, we have introduced the concept of a \emph{classical surrogate} for a quantum learning model. Having access to a classical replacement that can be efficiently constructed from a trained quantum learning model greatly enhances its applicability and interpretability. 
We have shown that a widely analyzed type of re-uploading models has a classical surrogate. This is possible because this class of re-uploading models can be expanded in terms of a truncated Fourier series with a modest number of coefficients, the classical surrogate is then also a truncated Fourier series whose coefficients can be found efficiently by performing a discrete Fourier transform. 

Classical surrogates have utility beyond removing the need of a quantum device for their use in production environments. 
They provide a natural benchmark by offering a concrete and natural test for any claim of \enquote{quantum advantage}: A quantum learning model can not exhibit a quantum advantage if it does not possess trainability, expressivity or generalization properties superior to its classical surrogate. Therefore, classical surrogates can be used as a tool to pin down regimes where a possible quantum advantage could occur by indicating when the quantum model enters a \enquote{classical regime} where one could just equivalently train the classical surrogate. Conversely, impracticality of the classical surrogate can indicate a possible regime of advantage.

Applying these concepts to selected simple learning problems, we have observed that re-uploading models constructed from contemporary circuit ansätze can be beat by a simple classical surrogate model as we did not witness advantages in training, performance or an inductive bias towards favorable solutions. It is still conceivable that such an inductive bias could exist, but our understanding of its relation to particular circuit templates is too ill-understood to realize it. 
It is our hope that this work stimulates further research into the precise potential of variational quantum circuit for learning tasks and into classical surrogates for other classes of quantum learning models.

\emph{Acknowledgments.}
We would like to thank Sofiene Jerbi, Ingo Roth and Daniel Stilck-França for insightful discussions. We thank the BMBF (Hybrid), the BMWK (PlanQK), the QuantERA (HQCC), the Munich Quantum Valley (K8) and the Einstein Foundation (Einstein Research Unit on Quantum Devices) for their support.

\bibliography{main}
\clearpage

\setcounter{secnumdepth}{2}
\onecolumngrid
\appendix

\section{Recovery guarantees}\label{sec:app_recovery_guarantees}

In our reconstruction of the Fourier coefficients, we have to take into account that we have finite sampling statistics when we compute expectation values from the outputs of the quantum device. To do so, we first start by proving a concentration bound for a vector of sample means.
\begin{lemma}\label{lemma:one_norm_concentration}
Let $\hat\xxi = ( \hat\xi_{i,j} )_{i=1}^{T}\mathstrut_{j=1}^N \in \calX$ be a collection of i.i.d.\ zero-centered random variables such that $|\hat\xi_{i,j}| \leq B$ for all $i$ and $j$. Let $\hat\eeta$ denote the vector of the $T$ sample means $\hateta_i = \frac{1}{N}\sum_{j=1}^N \hat\xi_{i,j}$. The $\ell_1$-norm of $\hat\eeta$ then obeys the large deviation bound
\begin{align}
\bbP\left[ \lVert \hat\eeta \rVert_1 \geq \alpha \right] \leq  \exp\left(\log(2) T -\frac{\alpha^2 N}{2 T B^2} \right).
\end{align}
\end{lemma}
Before we proceed to the proof of the above Lemma, we note that this is an improvement over what could be obtained by element-wise application of Hoeffding's inequality which would yield a right hand side of
\begin{align}
    \exp\left(\log(2) \log(2T) -\frac{\alpha^2 N}{2T^2 B^2}\right).
\end{align}

The proof of Lemma~\ref{lemma:one_norm_concentration} relies on Gaussian concentration inequalities obtained from transportation cost inequalities. To state the main theorem we build on, we first need a definition:
\begin{definition}
A probability measure $\mu$ on a space $\calX$ with distance measure $d$ satisfies $T_1(c)$ if the 1-Wasserstein distance of $\mu$ to any other measure $\nu$ on $\calX$,
\begin{align}
    W_1(\mu, \nu) = \inf \left\{ \int_{\calX \times \calX} \diff \pi(x,y) \, d(x,y) \, \left| \vphantom{\int_{\calX}}\right. \, \int_{\calX} \diff y \, \pi(x,y) = \mu(x), \int_{\calX} \diff x \, \pi(x,y) = \nu(y) \right\},
\end{align}
obeys an upper bound through the relative entropy of the form
\begin{align}
    W_1(\mu, \nu) \leq \sqrt{2 c D(\nu \lVert \mu)}.
\end{align}
\end{definition}
A measure that fulfills $T_1(c)$ has nice concentration properties as is witnessed by the following Theorem:
\begin{theorem}[Corollary 3.4.1 of Ref.~\cite{raginsky2015concentration}]\label{thm:gauss_concentration}
Let $\mu$ be a probability distribution over $\xxi \in \calX$ that satisfies $T_1(c)$ and $F$ be a Lipschitz function with Lipschitz constant $\lVert F \rVert_{\mathrm{Lip}}$. Then
\begin{align}
    \bbP\left[ F(\xxi) - \bbE[ F(\xxi) ] \geq \alpha \right] \leq \exp\left(-\frac{\alpha^2}{2 c \lVert F \rVert_{\mathrm{Lip}}^2} \right).
\end{align}
\end{theorem}
Before we can come to the final proof, we establish that the type of measure we are interested in fulfills $T_1(c)$ due to the bounded nature of the involved variables:
\begin{lemma}\label{lemma:TC1}
    Let $\{ \mu_i \}_{i=1}^n$ be measures on the interval $\calI = [-B, B]$ equipped with the distance measure $d(x,y) = |x-y|$. Then, $\bigotimes_{i=1}^n \mu_i$ as a measure over $\calI^n$ is $T_1( n B^2)$ with respect to $d(x_n, y_n) = \lVert x_n - y_n \rVert_1$.
\end{lemma}
\begin{proof}
We can upper-bound the 1-Wasserstein distance for $\mu$ to any other measure on $\calI$ through the diameter of $\calI$~\cite{gibbs2002choosing}
\begin{align}
    W_1(\mu, \nu) \leq \operatorname{diam}(\calI) d_{\text{TV}}(\mu, \nu).
\end{align}
Combining this with Pinsker's inequality yields
\begin{align}
    W_1(\mu, \nu) \leq \operatorname{diam}(\calI) \sqrt{\frac{1}{2} D(\nu \lVert \mu)} = \sqrt{2 B^2 D(\nu \lVert \mu)}
\end{align}
and hence, $\mu$ is $T_1(B^2)$. Then, we can use the tensorization of the $T_1$ property~\cite[Proposition 3.4.4]{raginsky2015concentration}
\begin{align}
    \mu_i = T_1(c) \text{ for all } i \Rightarrow \bigotimes_{i=1}^n \mu_i = T_1(nc)
\end{align}
to conclude the stated Lemma.
\end{proof}

We now have all the tools at hand to present the proof:
\begin{proof}[Proof of Lemma~\ref{lemma:one_norm_concentration}]
We start by defining a class of functions. Let $\sigma \in \{ -1, 1\}^T$ and define
\begin{align}
    F_{\sigma}(\hat\xxi) = \sum_{i=1}^T \sigma_i \frac{1}{N}\sum_{j=1}^N \hatxi_{i,j}.
\end{align}
We use this as a proxy for the $\ell_1$-norm as for all realizations $\eeta$ there exists a choice $\sigma_{+}$ such that $\lVert \eeta \rVert_1 = F_{\sigma_{+}}(\xxi)$. Also note that $\bbE[F_{\sigma}(\hat\xxi)] = 0$ for all $\sigma$.
The map $F_{\sigma}$ is Lipschitz with respect to the $\ell_1$-norm on $\calX$ for all $\sigma$ with Lipschitz constant
\begin{align}
    \lVert F_{\sigma} \rVert_{\text{Lip}} &= \sup_{\xxi \neq \xxi' \in \calX} \frac{| \sum_{i=1}^T \sigma_i \frac{1}{N}\sum_{j=1}^N \xi_{i,j}  - \sum_{i=1}^T \sigma_i  \frac{1}{N}\sum_{j=1}^N \xi_{i,j}' | }{\sum_{i=1}^T \sum_{j=1}^T |\xi_{i,j} - \xi_{i,j}'|} \\
    &= \frac{1}{N}\sup_{\xxi \neq \xxi' \in \calX} \frac{| \sum_{i=1}^T \sigma_i \sum_{j=1}^N (\xi_{i,j}  - \xi_{i,j}') | }{\sum_{i=1}^T \sum_{j=1}^T |\xi_{i,j} - \xi_{i,j}'|} \\
    &\leq \frac{1}{N} \sup_{\xxi \neq \xxi' \in \calX} \frac{ \sum_{i=1}^T \sum_{j=1}^N |\xi_{i,j}  -   \xi_{i,j}' | }{\sum_{i=1}^T \sum_{j=1}^T |\xi_{i,j} - \xi_{i,j}'|}
    \\
     &= \frac{1}{N},
\end{align}
where we have applied the triangle inequality. 
Applying Theorem~\ref{thm:gauss_concentration} to $F_{\sigma}$ and the underlying distribution of $\hat\xxi$ yields
\begin{align}
    \bbP\left[ F_{\sigma}(\hat\xxi) \geq \alpha \right] \leq \exp\left(-\frac{\alpha^2 N}{2 T B^2} \right),
\end{align}
irrespective of the particular choice of $\sigma$.
We now split the parameter space $\xxi \in \calX$ into parts $\calX_{\sigma}$ such that
\begin{align}
    F_{\sigma}(\xxi) = \lVert \eeta \rVert_1 \text{ for all } \xxi \in \calX_{\sigma},
\end{align}
and note that $\calX = \bigcup_{\sigma} \calX_{\sigma}$ and that $\calX_{\sigma} \cap \calX_{\sigma'}$ has measure zero for $\sigma \neq \sigma'$. With this, we can now conclude that
\begin{align}
    \bbP[\lVert \hat\eeta \rVert_1 \geq \alpha] &= \bbP[\max_{\sigma'} F_{\sigma'}(\hat\xxi)] \geq \alpha] \\
    &= \sum_{\sigma} \bbP[\max_{\sigma'} F_{\sigma'}(\hat\xxi) ] \geq \alpha \, \cap \,  \hat\xxi \in \calX_{\sigma}]\\
    &= \sum_{\sigma} \bbP[F_{\sigma}(\hat\xxi) \geq \alpha \, \cap \,  \hat\xxi \in \calX_{\sigma}] \\
    &\leq \sum_{\sigma} \bbP[F_{\sigma}(\hat\xxi) \geq \alpha] \\
    &\leq 2^T \exp\left(-\frac{\alpha^2 N}{2 T B^2} \right).
\end{align}
Bringing the prefactor into the exponent concludes the proof.
\end{proof}

To complete the guarantees for the reconstruction we need bounds on the largest and smallest non-zero eigenvalue of the matrix relevant for the reconstruction which is given by a Vandermonde matrix. We will make use of the fact that discrete Fourier transforms with equally spaced sampling points are ideally conditioned:
\begin{lemma}[See Ref.~\cite{nagel2020condition}]\label{lem:cond_dft}
All singular values of the Vandermonde matrix $A \in \bbC^{T \times T}$ with entries $A_{j,k} = e^{-2 \pi i \frac{j,k}{T}}$, $j,k \in \{0, 1, \dots, T-1\}$, are equal to $\sqrt{T}$.
\end{lemma}
This Lemma underpins performance guarantees for a univariate discrete Fourier transform. If we have $x \in [0, 2\pi)$ and $\omega_k \in \Omega = \{ - \Omega_0, \dots, \Omega_0\}$ we set $T = |\Omega| = 2 \Omega_0 + 1$ and $x_j = 2\pi j / T$ so that 
\begin{equation}
A_{j,k}' = e^{-i x_j \omega_k} = e^{-2 \pi i \frac{j,k}{T}}e^{-2\pi \frac{j \Omega_0}{T}}
\end{equation}
reproduces the above formula up to a phase shift. This phase shift accounts for the difference of discrete Fourier transform and centered Fourier transform and does not alter the singular values as it can be seen as $A' = D A$ for a diagonal matrix with phases on the diagonal such that $D^{\dagger} D = \bbI$ and hence $A'^{\dagger} A' = A^{\dagger} D^{\dagger} D A = A^{\dagger} A$. 

The multivariate extension of this is given by performing a DFT on a grid generated by the appropriate DFTs on each coordinate with the corresponding Vandermonde matrix being the tensor product of the individual Vandermonde matrices. This yields the following multivariate corollary of the above Lemma:
\begin{corollary}\label{corr:A_singular_values}
Consider a set of frequency vectors $\oomega \in \Omega$ such that the maximal frequency in every coordinate is $\omega_{\max}(i) = \max \{ | \omega_i | \colon \oomega \in \Omega \}$. Then using a grid generated by choosing $T_i = 2 \omega_{\max}(i) + 1$ equally spaced values in the interval $x_i \in [0, 2\pi)$ and performing the discrete Fourier transform on the product (over $T = \prod_i T_i$ points) yields a Vandermonde matrix $A$ with all singular values equal to
\begin{align}
    \sqrt{T} = \sqrt{\prod_{i=1}^d T_i} = \sqrt{\prod_{i=1}^d [ 2 \omega_{\max}(i) + 1]}.
\end{align}
\end{corollary}

With this we can now deliver the proof of Proposition~\ref{prop:recovery_guarantees} of the main text:
\begin{proposition}[\ref{prop:recovery_guarantees}]
Let $\cc_{*}$ be the vector of Fourier coefficients obtained by performing the protocol outlined in the main text with a number of samples per datapoint $N$. We can guarantee
\begin{align}
    \bbP\left[ \sup_{\xx \in \calX} \lvert f(\xx) - g_{\cc_{*}}(\xx)  \rvert \leq \epsilon \right] \geq 1 - \delta
\end{align}
if
\begin{align}
    N = \frac{2 \lVert M \rVert_{\infty}^2}{\epsilon^2}\left( \log \frac{1}{\delta} + T \log 2 \right) 
\end{align}
and hence perform a total of
\begin{align}
    N_{\mathrm{total}} = T N = \frac{2 T \lVert M \rVert_{\infty}^2}{\epsilon^2}\left( \log \frac{1}{\delta} + T \log 2 \right) 
\end{align}
invocations of the quantum learning model.
\end{proposition}
\begin{proof}
We first use the fact that
\begin{align}
    \sup_{\xx \in \calX} \lvert f(\xx) - g_{\cc_{*}}(\xx)  \rvert &= \sup_{\xx \in \calX} \left\lvert \sum_{\oomega \in \Omega} (c_{\oomega} - c_{*,\oomega})e^{-i \oomega \xx}  \right\rvert \\
    &\leq \sum_{\oomega \in \Omega} \sup_{\xx \in \calX}  |c_{\oomega} - c_{*,\oomega}| |e^{-i \oomega \xx}|
    \\ &= \lVert\cc - \cc_{*} \rVert_1.
\end{align}
Note that this bound can not be improved without further assumptions on $f$ and $g$. 

We obtain $\cc_{*}$ by applying a discrete Fourier transform to every data feature separately. Generically, we can write the relation between the outputs of a Fourier series $\yy = \{ f(\xx_i) \}_{i=1}^T$ and the Fourier coefficients $\cc$ as $\yy = A \cc$ where $A$ is a Vandermonde-type matrix isomorphic to the tensor product of the local Vandermonde matrices associated to the discrete Fourier transforms on the different data features $A_i$. A particularly nice property of the discrete Fourier transform is that $A$ is invertible in this case.

Note that experimentally, we use a sample mean estimate to approximate the output of the quantum model, hence we obtain an estimator $\hat\yy$. Because of the linear structure, we can decompose our estimate into a perfect term and an error term
\begin{align}
    \hat\yy = A \cc + \hat\eeta.
\end{align}
Our estimate for the underlying Fourier coefficients is then given by the least-squares estimator 
\begin{align}
    \cc_{*} = A^{-1} \hat\yy,
\end{align}
The $\ell_1$-norm difference then becomes
\begin{align}
    \lVert\cc - \cc_{*} \rVert_1 &= \lVert \cc - A^{-1} (A \cc + \hat\eeta) \rVert_1 
    = \lVert A^{-1} \hat\eeta \rVert_1 
    = \frac{1}{\sqrt{T}}\lVert  \hat\eeta \rVert_1,
\end{align}
where we have exploited Corollary~\ref{corr:A_singular_values} and the fact that all singular values of $A$ are equal to $\sqrt{T}$ which implies the last equality.
To obtain a faithful estimate, we need to control the $\ell_1$-norm of the estimation error. To do so, we note that 
\begin{align}
    \bbP\left[ \sup_{\xx \in \calX} \lvert f(\xx) - g_{\cc_{*}}(\xx)  \rvert \geq \epsilon \right] \leq \bbP\left[ \lVert A^{-1} \eeta  \rVert_1 \geq \epsilon \right].
\end{align}
as the latter event implies the other. As
\begin{align}
    \bbP\left[ \lVert A^{-1} \eeta  \rVert_1 \geq \epsilon \right] = \bbP\left[ \lVert \eeta  \rVert_1 \geq \sqrt{T}\epsilon \right],
\end{align}
we can now apply Lemma~\ref{lemma:one_norm_concentration} to the random variable $\hat\eeta$ with $\alpha = \sqrt{T}\epsilon$, $B = \lVert M \rVert_{\infty}$ to obtain
\begin{align}
    \bbP\left[ \sup_{\xx \in \calX} \lvert f(\xx) - g_{\cc_{*}}(\xx)  \rvert \geq \epsilon \right] \leq
    \exp\left(\log(2) T -\frac{\epsilon^2 N}{2 \lVert M \rVert_{\infty}^2} \right).
\end{align}
Setting the right hand side equal to $\delta$ and solving for $N$ yields the statement of the Proposition.
\end{proof}

\section{Sample complexity of estimating multiple observables}\label{sec:app_multiple_observables}

Obtaining the outputs of a quantum learning model defined via expectation values also comes with a overhead. For the sake of completeness, we give a proof of the well-known sample complexity, compare e.g.\ Ref.~\cite{kliesch2021theory}.
\begin{lemma}\label{lemma:max_norm_concentration}
Let $\hat\xxi = ( \hat\xi_{i,j} )_{i=1}^{T}\mathstrut_{j=1}^N \in \calX$ be a collection of i.i.d.\ zero-centered random variables such that $|\hat\xi_{i,j}| \leq B$ for all $i$ and $j$. Let $\hat\eeta$ denote the vector of the $T$ sample means $\hateta_i = \frac{1}{N}\sum_{j=1}^N \hat\xi_{i,j}$. We can guarantee
\begin{align}
\bbP\left[ \lVert \hat\eeta \rVert_{\infty} \geq \epsilon \right] \leq \delta
\end{align}
for a total number of i.i.d.\ copies
\begin{align}
    N_{\mathrm{total}} = NT \geq \frac{2 B^2}{\epsilon^2} T \log\frac{2 T}{\delta}.
\end{align}
\end{lemma}
\begin{proof}
    For every entry of $\hat\eeta$, we have by Hoeffding's inequality that
    \begin{align}
        \bbP\left[ \lvert \hat\eta_i \rvert \geq \epsilon \right] \leq 2 \exp\left(- \frac{N \epsilon^2}{2 B^2}\right).
    \end{align}
    We can then use the union bound to conclude
    \begin{align}
        \bbP\left[ \lVert \hat\eeta \rVert_{\infty} \geq \epsilon \right] &= \bbP\left[ \bigcup_{i=1}^T \{ \lvert \hat\eta_i \rvert \geq \epsilon \} \right] \\
        &\leq \sum_{i=1}^T \bbP\left[  \lvert \hat\eta_i \rvert \geq \epsilon  \right]\\
        &= \sum_{i=1}^T 2 \exp\left(- \frac{N \epsilon^2}{2 B^2}\right)\\
        &= 2T \exp\left(- \frac{N \epsilon^2}{2 B^2}\right).
    \end{align}
    Equating the right hand side to $\delta$ and solving for $N$ yields the claim of the Lemma.
\end{proof}

\section{Discussion of ill-conditioning}\label{sec:app_additional_info_numerics}

When working with datasets, one does not have control over the 
datapoints for which we have corresponding labels, leading to a potentially ill-conditioned situation. This may arise when data points are too close to each other.
To illustrate this, consider the case where the Fourier matrix $A$ is square. The condition number of $A$, $\kappa(A)$, is given as the ratio of largest and smallest singular value. Since $\kappa(A)=\kappa(A^T)$, we can consider $A^T$ instead of $A$, such that the $i$'th column of $A^T$, $\aa_i$, is associated with data point $\xx_i$. The variational formulation of singular values gives 
for the condition number
\begin{align}
    \kappa(A^T) &= \frac{\max_{\norm{\zz_1}_2=1} \norm{A^T\zz_1}_2}{\min_{\norm{\zz_2}_2=1}\norm{A^T\zz_2}_2}.
\end{align}
Choosing $\zz_1 = (\ee_i + \ee_j)/\sqrt{2}$ and $\zz_2 = (\ee_i - \ee_j)/\sqrt{2}$, where $\ee_i$ is the $i$-th standard basis vector, yields the lower bound
\begin{align}
    \kappa(A^T) &\geq \frac{\norm{\aa_i + \aa_j}_2}{\norm{\aa_i - \aa_j}_2},
\end{align}
which clearly blows up when two column vectors approach each other, which happens exactly when two datapoints are too close.

\section{Additional information on numerical experiments}\label{sec:app_additional_figures}

In this section, we collect further information and interpretation for the numerical experiments. 

\subsection{Example of a trainable block}\label{sec:app_trainable_block}

An example of a trainable block $W^{(l)}$ of the \textit{Strongly Entangling Layer}
type discussed in the main text for three qubits is shown in Fig.~\ref{fig:w_architecture}.

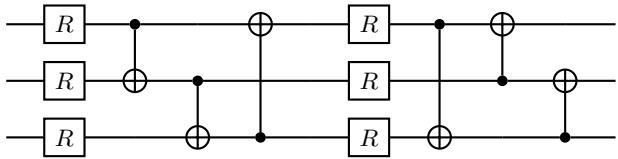
\begin{figure}	
    \centering
    \begin{quantikz}[row sep={0.75cm,between origins}]
        & \gate{R } & \ctrl{1} & \qw & \targ{}& \qw & \gate{R} & \ctrl{2} & \targ{}& \qw & \qw & \\
        & \gate{R} & \targ{}& \ctrl{1} & \qw & \qw & \gate{R} & \qw & \ctrl{-1} & \targ{}& \qw &  \\
        & \gate{R} & \qw & \targ{}& \ctrl{-2} & \qw & \gate{R} & \targ{}& \qw & \ctrl{-1} & \qw & 
    \end{quantikz}
    \caption{Example implementation of a trainable circuit block $W^{(l)}$ for three qubits with two block layers. The $R$-gates are arbitrary rotation gates of the form ${R(\alpha, \beta, \gamma)=R_X(\alpha)R_Z(\beta)R_X(\gamma)}$. In the $b$-th block layers, the CNOT gates connect $i$ and $(i+b) \mod n$, where $n$ is the number of qubits.}
    \label{fig:w_architecture}
\end{figure}

\subsection{Loss curves for numerical examples}

Fig.~\ref{fig:main_figure} shows the loss curves for the three examples discussed in the main text. For completeness, we show the results from Fig.~\ref{fig:synthetic} in the first row again.

\begin{figure*}
    \centering
    \includegraphics{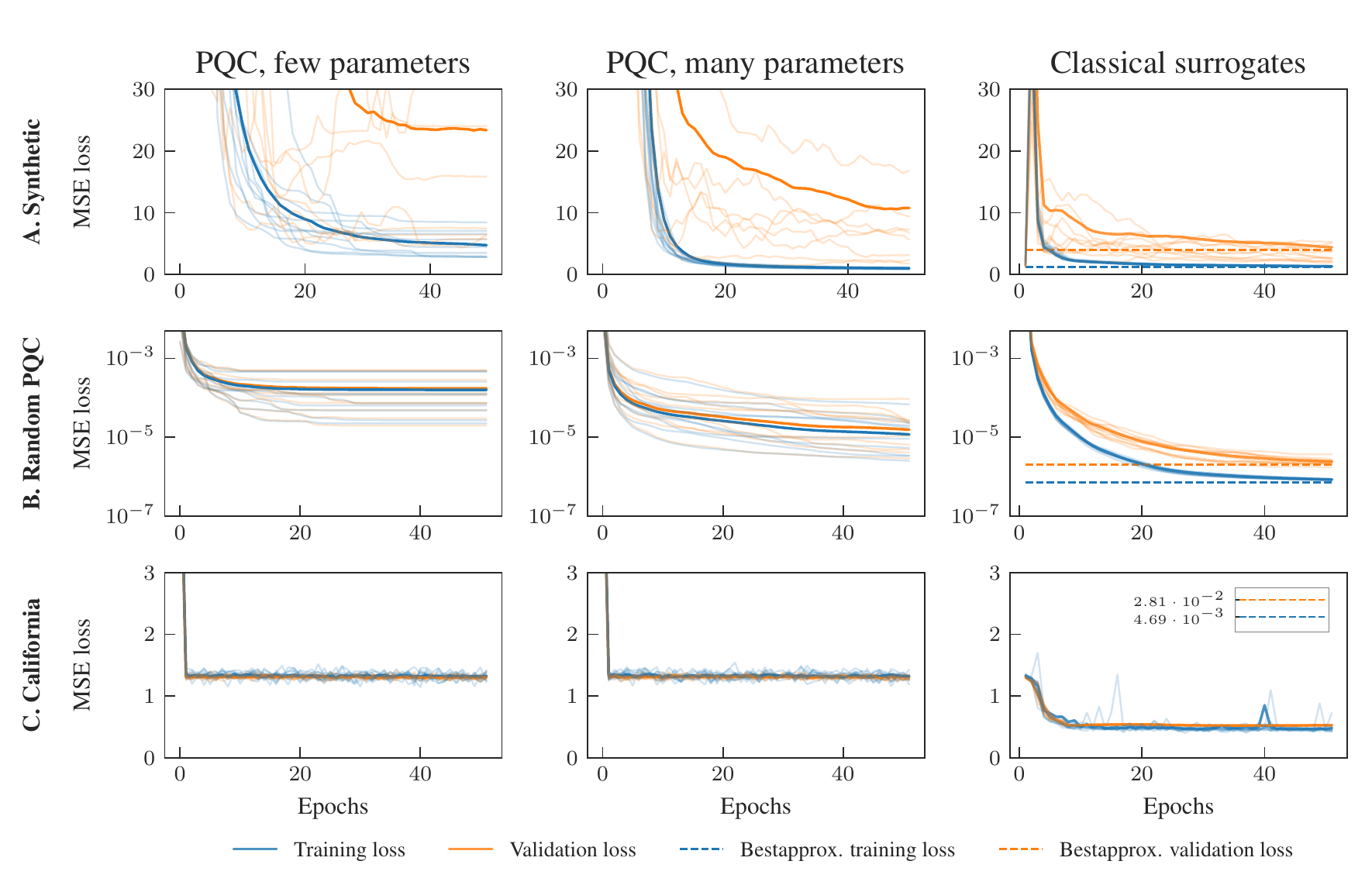}
    \caption{(\textbf{Row-wise}) In the top row the synthetic dataset (models have $L=2$), in the middle row the dataset sampled from a randomly initialized PQC (with $L=2$, $B=2$, the models have $L=3$) and in the bottom row the California housing dataset (models have $L=3$). (\textbf{Column-wise}) In the left column are quantum models where the number of block layers is $B=1$, \ie the number of trainable parameters $\ttheta$ is low. Depicted in the middle row are quantum models with $B=3$, resulting in a higher number of trainable parameters $\ttheta$. In the right column are the corresponding classical models. For the two smaller datasets the best approximation to the training data was computed by directly solving the linear least square problem. Note that the "best approximation validation loss" line only gives the validation loss corresponding to the lowest possible training loss, not the lowest possible validation loss. (\textbf{General}) All loss curves of the quantum models are the averages $N=50$ training runs, each with randomly initialized weights (darker color). For each model, ten individual runs were plotted for illustration.}
    \label{fig:main_figure}
\end{figure*}

\subsection{Example of random parametrized quantum circuits}\label{sec:app_pqc_learning_problem}

As a further numerical example, we use a synthetic dataset which is \enquote{natural} for the quantum models we consider. The learning problem is to predict the output of a randomly initialized re-uploading model of the structure introduced at the beginning of this section with input dimension $\xx_j \in\bbR^4$, $L=2$ and $B=2$. We compute $N=3500$ random samples.

In the second row of Fig.~\ref{fig:main_figure}, training trajectories are shown for two quantum learning models with $L=3$ and $B \in \{1,3\}$ and the classical surrogate for this problem. Here, we observe that even the smallest quantum learning model we analyze which has only one block layer has very smooth training curves. 
We observe that adding more block layers improves the performance of the quantum learning model, but it is always significantly less accurate than the classical surrogate.
The good performance of the classical surrogate is completely expected as the frequency structure of the generating process is contained in the frequency structure of the surrogate. Consequently, the optimization of the classical surrogate in this instance is more or less equal to the surrogation process introduced above, except that we did not have control over the data-points in this case which is alleviated by the fact that we have sufficiently many of them. Still, as this also holds true for the frequency structure of the quantum models, this results underlines how the convex loss landscape of the classical surrogate can lead to favorable trainability properties compared to the much more complex loss landscape of the quantum models.

All in all, we see that the observations made for the first dataset also hold true in this setting. As we increase the number of parameters for the quantum models, we see a transition from a rugged loss landscape (high standard deviation in the loss functions) towards increasingly behaving like their classical surrogates (again accompanied by a drop in standard deviation of the loss functions). Despite the fact that we chose a very natural problem for the quantum models in question, they are consistently outperformed by their classical surrogates.

\subsection{Example of the California Housing Dataset}\label{sec:california_housing_problem}

The California housing dataset is a canonical small benchmark regression problem from classical machine learning. The dataset consists of $N=20640$ samples with input dimension $\xx_j \in \bbR^8$. The task is to predict the value of houses in the price range of $\$15000  - \$500 000$ which we map to the interval $0.15-5$.

Looking at the third row of Fig.~\ref{fig:main_figure}, we observe that even for models with very few parameters, we are immediately drawn into a local minimum for both parameter counts of the quantum learning model.  We assume this is due to unfavorable loss landscapes for the quantum learning models for this particular dataset combined with the limited expressivity of the models in terms of how many Fourier coefficients were available. The classical surrogate, on the contrary, manages to reach an improved validation and training loss. It is curious that the training loss of the classical surrogate does not approach the value resulting from direct inversion of the problem which is indicated by dashed lines in the inset. This is likely due to numerical reasons resulting from the LBFGS solver we employ. Altogether, while we are not able to observe a transition from high variance and rugged loss landscape to lower variance regimes as with the other models, the classical surrogate still reaches better results than the corresponding quantum models.

\section{Author Contributions}

F.J.S.\ conceived and conducted the numerical experiments. J.J.M.\ envisioned the theoretical part of this work with support from F.J.S. J.E.\ supported research and development. All authors contributed to the writing of the manuscript. 

\section{Data Availability}

Code for implementations and data of the numerical experiments conducted in this work will be made available upon reasonable request.

\end{document}

%% file: commands.tex
\newcommand{\diff}{\ensuremath\mathrm{d}}

\newcommand{\norm}[1]{\left\lVert#1\right\rVert}

\providecommand{\myvec}[1]{\ensuremath{\boldsymbol{#1}}}

\providecommand{\aa}{\ensuremath{\myvec{a}}}

\providecommand{\cc}{\ensuremath{\myvec{c}}}

\providecommand{\ee}{\ensuremath{\myvec{e}}}

\providecommand{\xx}{\ensuremath{\myvec{x}}}
\providecommand{\yy}{\ensuremath{\myvec{y}}}
\providecommand{\zz}{\ensuremath{\myvec{z}}}

\providecommand{\eeta}{\ensuremath{\myvec{\eta}}}
\providecommand{\ttheta}{\ensuremath{\myvec{\theta}}}

\providecommand{\xxi}{\ensuremath{\myvec{\xi}}}

\providecommand{\oomega}{\ensuremath{\myvec{\omega}}}

\providecommand{\hateta}{\ensuremath{\hat{\eta}}}

\providecommand{\hatxi}{\ensuremath{\hat{\xi}}}

\providecommand{\calF}{\ensuremath{\mathcal{F}}}
\providecommand{\calG}{\ensuremath{\mathcal{G}}}

\providecommand{\calI}{\ensuremath{\mathcal{I}}}

\providecommand{\calS}{\ensuremath{\mathcal{S}}}

\providecommand{\calX}{\ensuremath{\mathcal{X}}}
\providecommand{\calY}{\ensuremath{\mathcal{Y}}}

\providecommand{\bbC}{\ensuremath{\mathbb{C}}}

\providecommand{\bbE}{\ensuremath{\mathbb{E}}}

\providecommand{\bbI}{\ensuremath{\mathbb{I}}}

\providecommand{\bbP}{\ensuremath{\mathbb{P}}}

\providecommand{\bbR}{\ensuremath{\mathbb{R}}}
